\documentclass[10pt,letterpaper]{article}

\usepackage[hang,perpage]{footmisc}

\usepackage{amsmath}
\usepackage{amsthm}
\usepackage{amssymb}
\usepackage{stmaryrd}
\usepackage{graphicx}
\usepackage{latexsym}
\usepackage{times}
\usepackage[usenames]{color}
\usepackage{braket}
\usepackage{textcomp}
\usepackage{makeidx}
\usepackage{tocbibind}
\usepackage[bookmarks=true,backref=true,colorlinks=true,linkcolor=darkolivegreen,
citecolor=darkolivegreen,urlcolor=darkolivegreen]{hyperref}
\numberwithin{equation}{subsection}
\theoremstyle{definition}
\newtheorem{ass}{Assumption}[section]
\newtheorem{theorem}[ass]{Theorem}
\newtheorem{lemma}[ass]{Lemma}
\newtheorem{definition}[ass]{Definition}
\newtheorem{corollary}[ass]{Corollary}

\newtheorem{ex}[ass]{Example}
\newtheorem{rem}[ass]{Remark}
\def\indexname{Index of terminology}
\makeindex

\newcommand{\captionfonts}{\footnotesize}
\makeatletter  
\long\def\@makecaption#1#2{%
  \vskip\abovecaptionskip
  \sbox\@tempboxa{{\captionfonts #1: #2}}%
  \ifdim \wd\@tempboxa >\hsize
    {\captionfonts #1: #2\par}
  \else
    \hbox to\hsize{\hfil\box\@tempboxa\hfil}%
  \fi
  \vskip\belowcaptionskip}
\makeatother   
\definecolor{darkolivegreen}{rgb}{0.333333, 0.419608, 0.1843140}
\allowdisplaybreaks

\makeatletter
\def\printnotation{{%
\def\indexname{Index of notation}
\begin{theindex}
\@input{\jobname.ntn}
\end{theindex}
}}
\makeatother

\makeglossary

\begin{document}


\title{On the stability of the massive scalar field 
in Kerr space-time}

\author{
Horst Reinhard Beyer \\
MPI for Gravitational Physics \\
Am Muehlenberg 1 \\ 
D-14476 Golm \\
Germany \\
Email: horst@aei.mpg.de}

\date{\today}                                     

\maketitle

\begin{abstract}

The current early stage in 
the investigation of the stability of the Kerr metric is 
characterized by the study of appropriate model 
problems. Particularly interesting is the problem 
of the stability of the solutions of the Klein-Gordon equation, 
describing the propagation of a scalar field in the background of a 
rotating (Kerr-) black hole. Results suggest that the stability of the 
field depends crucially on its mass $\mu$. Among others, 
the paper provides an 
improved bound 
for $\mu$ above which 
the solutions of the reduced, by separation in the azimuth angle in 
Boyer-Lindquist coordinates, 
Klein-Gordon equation are stable.  Finally, it gives new 
formulations of the reduced equation, in particular, in form 
of a time-dependent wave 
equation that is governed by a family of 
unitarily equivalent positive 
self-adjoint operators. The latter formulation 
might turn out useful for further 
investigation. On the other hand, it is proved that from 
the abstract properties
of this family alone it cannot be concluded that the corresponding 
solutions are stable. 
\end{abstract}

\section{Introduction}
\label{introduction}

Kerr space-time is the only possible vacuum exterior 
solution of Einstein's field equations describing 
a stationary, rotating, 
uncharged black hole with non-degenerate event horizon 
\cite{robinson} and is expected to be the unique, stationary, 
asymptotically flat, vacuum space-time containing a non-degenerate 
Killing horizon \cite{alexakis}. Also, it is expected 
to be the asymptotic limit of 
the evolution of asymptotically flat vacuum data in general relativity.
\newline
\linebreak 
An important step towards
establishing the validity of these expectations is the proof of 
the stability 
of
Kerr space-time. In comparison to Schwarzschild space-time, 
where linearized stability has been proved, 
this 
problem is complicated by a lower dimensional
symmetry group and the absence of a Killing field 
that is everywhere
time-like outside the horizon. For instance, the latter is reflected 
in the fact that energy densities corresponding to the Klein-Gordon field
in a Kerr gravitational field have no definite sign. 
This absence complicates the application
of methods from operator theory and of so called  
``energy methods'' that are both employed in estimating
the decay of solutions of hyperbolic partial differential 
equations.\footnote{For the first, see, for instance, 
\cite{beyer2}. For the second, see, for instance, 
Chapter~$2$ of \cite{morawetz}.}  
\newline
\linebreak
On the other hand, two facts are worth noting. For this, note that 
in the following any reference to coordinates implicitly assumes use 
of
Boyer-Lindquist coordinates \cite{boyer}. 
\newline
\linebreak
First, in addition to its Killing vector fields that generate
one-parameter groups of symmetries (isometries), 
Kerr space-time admits a Killing tensor  \cite{walker}
that is unrelated to its symmetries.
Initiated by his groundbreaking work \cite{carter}
on the complete separability 
of the Hamilton-Jacobi equation in a Kerr 
background, 
Carter discovered that an operator that is induced by this Killing tensor 
commutes with the wave operator. On the other hand, 
Carter's operator contains a second order
time derivative \cite{carter4}. An analogous
operator has been found for the operator governing
linearized gravitational perturbations of the Kerr geometry 
\cite{kalnins}. A recent study finds another such `symmetry 
operator' which only contains a first 
order time derivative and commutes with a rescaled wave 
operator \cite{beyer3}. Differently to Carter's
operator, this operator is analogous to symmetry operators
induced by one-parameter group of isometries of the metric, in that it 
induces a mapping in the data space that is compatible with time 
evolution, and therefore describes a true symmetry of the solutions.
It is likely that an analogous operator can be found for 
a rescaling of the linearized operator governing gravitational 
perturbations of the Kerr geometry. In case of existence, it should 
facilitate the generalization to a Kerr background of the 
Regge-Wheeler-Zerilli-Moncrief (RWZM) 
decomposition of fields on a 
Schwarz\-schild background \cite{regge,zerilli,moncrief,teukolsky,
press,kegeles} which in turn should greatly simplify
the analysis of the stability of Kerr space-time. 
\newline
\linebreak
Second, there is a Killing field that is 
time-like in an open neighborhood of the event 
horizon 
given by  
\begin{equation} \label{specialkillingfield} 
\xi := \partial_{t}  +  \frac{a}{2Mr_{+}} \, 
\partial_{\varphi} \, \, , 
\end{equation} 
where $\partial_{t},\partial_{\varphi}$ are coordinate vector
fields of Boyer-Lindquist coordinates corresponding to the coordinate 
time $t$ and the azimuthal angular coordinate $\varphi$, $M > 0$
is the mass of the black hole and $a \in [0,M]$ its rotational
parameter. Moreover, if
\begin{equation} \label{inequality}
\frac{a}{M} \leqslant \frac{\sqrt{3}}{3} \, \, ,
\end{equation}
$\xi$ is time-like in the ergoregion, see 
Lemma~\ref{timelikeregion}. On the other hand,  
$\partial_{t}$ itself is space-like in the ergoregion, null 
on the stationary 
limit surface and time-like outside. For these reasons,  
at least for $a$ satisfying (\ref{inequality}), it might 
be possible 
to ``join"
energy inequalities belonging to the Killing fields 
by $\xi$ and $\partial_{t}$.
\newline
\linebreak
The discussion of the stability of the Kerr black hole is 
in its early stages. The first intermediate goal is the proof or disproof 
of its stability under ``small'' perturbations. As mentioned before, 
the linearized 
stability of the Schwarzschild metric has already been proved. 
In that 
case, by using the RWZM decomposition of fields 
in a Schwarzschild background, the question of 
the stability can be 
completely reduced to the question of the stability of the solutions 
of the wave equation on Schwarzschild space-time. For  
Kerr space-time, a similar reduction is not known. If such 
reduction exists,  
there is no guarantee that the relevant equation is the 
scalar wave equation. It is quite possible that such 
equation contains 
an additional (even positive) potential term that, similar to the 
potential term introduced by a mass of the field, could result in 
instability of the solutions. Second, an
instability of a massive scalar field in a Kerr background 
could indicate instability of the metric against perturbations 
by matter which generically has mass. If this were the case, 
even a proof of the stability 
of Kerr space-time could turn out as 
a purely mathematical exercise with little
relevance for general relativity. 
Currently, the 
main focus is the study of  
the stability of the solutions of the Klein-Gordon field on 
a Kerr background with the hope that the results lead to insight 
into the problem of linearized stability. Although the results 
of this paper also apply to the case that $\mu = 0$, its 
main focus is the case of Klein-Gordon fields 
of mass $\mu > 0$.
\newline
\linebreak
Quite differently from the case of a Schwarzschild background, the
results for these test cases suggest an
asymmetry between the cases 
$\mu = 0$ and $\mu \neq 0$. In the case of 
the wave equation, i.e., $\mu = 0$,  
results point to the stability of the solutions 
\cite{whiting,finster,dafermos,andersson,krivan1,krivan2}, 
whereas for $\mu \neq 0$, there are a number of results 
pointing in the direction of instability of the solutions 
under certain conditions \cite{damour,detweiler,zouros,furuhashi,khanna,
cardoso,hod}.
\newline
\linebreak
In particular,  
unstable modes were found by the  
numerical investigations by 
Furuhashi and Nambu for $\mu M \sim 1$ and $(a/M) = 0.98$,
by Strafuss and Khanna for $\mu M \sim 1$ 
and $(a/M) = 0.9999$ and by Cardoso and Yoshida for 
$\mu M \leqslant 1$ and $0.98 \leqslant (a/M) < 1$. The 
analytical study by Hod and Hod finds
unstable modes for $\mu M \sim 1$ with a growth rate
which is four orders of 
magnitude larger than previous estimates. On the other hand, 
\cite{beyer1} proves
that the restrictions of 
the solutions of the separated, in the azimuthal coordinate, Klein-Gordon 
field (RKG) are stable for 
\begin{equation} \label{massineq}
\mu \geqslant \frac{|m|a}{2Mr_{+}} 
\sqrt{1 + \frac{2M}{r_{+}} + \frac{a^2}{r_{+}^2}}
\, \, .
\end{equation}
Here
$m \in {\mathbb{Z}}$ is the 
`azimuthal separation parameter' and $r_{+} := M + 
\sqrt{M^2 - a^2}$. So far, this has been the only mathematically 
rigorous result on the stability of the solutions of the 
RKG for $\mu > 0$. This result 
contradicts the result of Zouros and Eardley, 
but is 
consistent with 
the other results above. In addition, there is the numerical
result by Konoplya and Zhidenko, \cite{konoplya} which confirms 
the result of Beyer, {\it but also 
finds 
no unstable modes of the RKG for $\mu M \ll 1$ 
and $\mu M \sim 1$}. 
\newline
\linebreak
Among others, this paper improves the estimate (\ref{massineq}). 
It is proved that 
the solutions of the RKG are stable for $\mu$
satisfying 
\begin{equation*} 
\mu \geqslant \frac{|m|a}{2Mr_{+}} 
\sqrt{1 + \frac{2M}{r_{+}}}
\, \, .
\end{equation*}
Further, it gives new 
formulations for RKG, in particular, in form of a time-dependent 
wave equation that  is governed by a family of unitarily equivalent 
positive 
self-adjoint operators. The latter might turn out useful in future 
investigations. On the other hand, it is proved that from 
the abstract properties
of this family alone it cannot be concluded that the corresponding 
solutions are stable. 
\newline
\linebreak 
The remainder of the paper is organized as follows.
Section~$2$ gives the geometrical setting of the 
discussion of the solutions of the RKG and a proof of the 
above mentioned property of the Killing field $\xi$. 
Section~$3$ gives 
basic properties of operators read off 
from the equation, including some new results. These properties
provide the basis for a formulation 
of the initial-value problem for the equation in Section~$4$
which is less dependent on methods from semigroups of operators 
than that of \cite{beyer1}. Section~$4$ also contains the improved result
on the stability of the solutions of RKG, a formulation of the RKG 
in terms of a time-dependent wave equation and the above mentioned 
counterexample.
Finally, the paper concludes 
with a discussion of the results and $2$ appendices that contain 
proof of results that were omitted in the main text to improve the
readability of the paper.

\section{The Geometrical Setting}
\label{geometricalsetting}

In Boyer-Lindquist coordinates\footnote{If not 
otherwise indicated, the symbols 
$t,r,\theta,\varphi$ denote coordinate projections whose 
domains will be obvious from the context. In addtion, we assume 
the composition 
of maps, which includes addition, multiplication and so forth, 
always to be maximally defined. For instance, the sum of two 
complex-valued maps is defined on the intersection of their domains.
Finally, we use 
Planck units where 
the reduced Planck constant $\hbar$, the 
speed of light in vacuum $c$, 
and the gravitational constant $\gamma$, all have the numerical 
value $1$.},
$(t,r,\theta,\varphi) : 
\Omega \rightarrow {\mathbb{R}}^4$,
the Kerr metric $g$ is given by
\begin{align*}
& g = g_{tt} \, dt \otimes dt + g_{t\varphi} ( dt \otimes d\varphi 
+ d\varphi \otimes dt) + g_{rr} \, dr \otimes dr + g_{\theta \theta} \,
 d\theta \otimes d\theta +  g_{\varphi \varphi} \, 
 d\varphi \otimes d\varphi \, \, , 
\end{align*}
where 
\begin{align*}
& g_{tt} := 1 - \frac{2Mr}{\Sigma} \, \, , \, \,
 g_{t\varphi} := \frac{2Mar\sin^2\!\theta}{\Sigma} \, \, , \, \,
g_{rr} := - \frac{\Sigma}{\Delta} \, \, , \, \, 
g_{\theta \theta} :=  -\Sigma \, \, ,  \\
& g_{\varphi \varphi} := - \frac{\Delta \overline{\Sigma}}{\Sigma} 
\sin^2\!\theta \, \, , 
\end{align*}
$M$ is the mass of the black hole, $a \in [0,M]$ is the rotational
parameter and
\begin{align*} 
& \Delta := r^2 - 2Mr +a^2 \, \, , \, \, \Sigma := r^2 + a^2 
\cos^2\!\theta \, \, , \\
& 
\overline{\Sigma} := 
\frac{(r^2+a^2) \Sigma + 2 M a^2 r \sin^2\!\theta}{\Delta} 
= \frac{(r^2+a^2)^2}{\Delta} - a^2 \sin^2\!\theta = 
\Sigma + 2 M r + \frac{4M^2r^2}{\Delta} \, \, , \\
& r_{+} := M + \sqrt{M^2-a^2} \, \, , \, \,  r_{-} := M - \sqrt{M^2-a^2} 
\, \, , \, \, \\
& \Omega := {\mathbb{R}} \times (r_{+},\infty) \times
 (0,\pi) \times (-\pi,\pi) \, \, .
\end{align*}

In these coordinates, the reduced 
Klein-Gordon equation corresponding to 
$m \in {\mathbb{Z}}$, 
governing solutions 
$\psi : \Omega \rightarrow {\mathbb{C}}$ of the form 
\begin{equation*}
\psi(t,r,\theta,\varphi) = \exp(im\varphi) \, u(t,r,\theta) \, \, ,
\end{equation*} 
where $u : 
\Omega_{s} \rightarrow {\mathbb{C}}$,
\begin{equation*}  
\Omega_{s} := (r_{+},\infty) \times (0,\pi) \, \, ,
\end{equation*}
for all  
$t \in {\mathbb{R}}$, $\varphi \in (-\pi,\pi)$, 
$(r,\theta) \in \Omega_{s}$, is given by

\begin{align} \label{kleingordonequation}
& \frac{\partial^2 u}{\partial t^2} + i b \, 
\frac{\partial u}{\partial t} + D^2_{r \theta} \, u
= 0 \, \, ,
\end{align}
where
\begin{equation*}
b := 
\frac{4mMar}{\triangle\overline{\Sigma}} = 
\frac{4mMar}{(r^2+a^2)^2 - a^2 \triangle 
\sin^2 \theta } =
 \frac{4mMar}{(r^2+a^2) \Sigma + 2 M a^2 r \sin^2\!\theta} \, \, ,
\end{equation*}
\begin{align*} \label{D2rtheta}
& D^2_{r \theta} f := \frac{1}{\overline{\Sigma}} 
\left( - \frac{\partial}{\partial r} \triangle  \frac{\partial}{\partial r}
- \frac{m^2 a^2}{\triangle} - \frac{1}{\sin \theta} \,
\frac{\partial}{\partial \theta} \sin \theta \,\frac{\partial}{\partial \theta}
+ \frac{m^2}{\sin^2 \theta} + \mu^2 \Sigma \right)f 
\end{align*} 
for every $f \in C^2(\Omega_{s},{\mathbb{C}})$ and 
$\mu \geq 0$ is the mass of the field. In particular, note that 
$b$ defines a real-valued bounded function on $\Omega_s$ which 
positive for $m \geq 0$ and negative for $m \leq 0$.
For this reason, it induces a bounded self-adjoint 
(maximal multiplication) operator $B$ on the weighted 
$L^2$-space $X$, see below, which is positive for $m \geq 0$ 
and negative for $m \leq 0$.  
Further, $D^2_{r \theta}$ 
is singular since the continuous extensions of the coeffcients 
of its highest (second) order radial derivative vanish on the horizon 
$\{r_{+}\} \times [0,\pi]$.
\newline
\linebreak
In particular, the following proves that the Killing field 
\begin{equation*} 
\xi := \partial_{t}  +  \frac{a}{2Mr_{+}} \, 
\partial_{\varphi} 
\end{equation*} 
is time-like in an open neighborhood of the event 
horizon and time-like in the ergoregion
if 
\begin{equation*} 
\frac{a}{M} \leqslant \frac{\sqrt{3}}{3} \, \, .
\end{equation*}
Proofs are given in Appendix~$1$.

\begin{lemma} \label{killingfields}
Let $M > 0, a > 0$.
For every $s \in {\mathbb{R}}$, the function 
\begin{equation*}
 g(\partial_{t} + s \, \partial_{\varphi},
\partial_{t} + s \,  \partial_{\varphi})
\end{equation*}
has a continuous extension to ${\overline{\Omega}}_{s}$. This 
extension is positive on  $\partial {\overline{\Omega}}_{s}$
if and only if 
\begin{equation*}
s = \frac{a}{2Mr_{+}} \, \, .
\end{equation*}
Further, 
\begin{equation*}
\xi := \partial_{t} + \frac{a}{2Mr_{+}} \, \partial_{\varphi}
\end{equation*}
is time-like precisely on 
\begin{equation*}
\Omega_{e2} := 
\left[2 M r_{+} - a^2 \sin^2\!\theta - a \, \Delta^{1/2} \sin \theta  
\left(1 + \frac{2M}{r-r_{-}} \right) 
\right]^{-1}(\,(0,\infty)\,) \, \, .
\end{equation*}
\end{lemma}

\begin{proof} See Appendix~$1$.
\end{proof}

\begin{lemma} \label{timelikeregion}
Let $M > 0, a > 0$ and $\Omega_{e1}$, defined by  
\begin{equation*}
\Omega_{e1} := (a^2 \sin^2\!\theta -\triangle)^{-1}((0,\infty))
\, \, , 
\end{equation*}
denote the ergoregion. If
\begin{equation} \label{condition}
\frac{a}{M} \leq \frac{\sqrt{3}}{3} \, \, ,
\end{equation} 
then 
\begin{equation*}
\Omega_{e1} \subset \Omega_{e2} \, \, .
\end{equation*}
\end{lemma}

\begin{proof} See Appendix~$1$.
\end{proof}

\section{Basic Properties of Operators in 
the Equation}
\label{basicproperties}

In a first step, we represent  
(\ref{kleingordonequation}) as a differential equation for 
an unknown function $u$ with values in   
a Hilbert space. For this reason, we represent 
formal operators present in 
(\ref{kleingordonequation}) as operators with well-defined 
domains in an appropriate Hilbert space and, subsequently, study basic
properties of the resulting operators. 
Theorems~\ref{core},~\ref{friedrichsextension} provide new results.  
 
\begin{definition}
In the following, $X$ denotes the weighted $L^2$-space $X$ defined by   
\begin{equation} \label{definitionofX}
X := L_{C}^2\big(\Omega_{s}\,,\overline{\Sigma}
\sin \theta\big) \, \, .
\end{equation} 
Further, $B$ is the bounded linear self-adjoint 
operator
on $X$ given by 
\begin{equation} \label{definitionofB}
B f := b f
\end{equation}
for every $f \in X$. Note that $B$ is positive for $m \geq 0$ and 
negative for $m \leq 0$.
\end{definition}

\begin{rem}
We note that, as consequence of the fact that 
$B \in L(X,X)$ is self-adjoint, the operator 
\begin{equation*}
\exp((it/2) B) \, \, , 
\end{equation*} 
where $\exp$ denotes the exponential function on 
$L(X,X)$, see, e.g., Section~3.3 in \cite{beyer2}, is unitary
for every $t \in {\mathbb{R}}$ and coincides
with the maximal multiplication operator by the function
$\exp((it/2) b)$. 
\end{rem}

\begin{definition} ({\bf Definition of $A_0$})
\begin{itemize}
\item[(i)]
We define $D(A_0)$ to consist of 
all $f \in C^2({\bar{\Omega}}_{s},{\mathbb{C}}) \cap X$ 
satisfying the conditions a), b) and c):    
\begin{itemize}
\item[a)] $D^2_{r \theta} f \in X$, 
\item[b)]
there is $R > 0$ such that $f(r,\theta) = 0$ for all 
$r > R$ and $\theta \in I_{\theta} := (0,\pi)$, 
\item[c)]  
\begin{equation*}
\lim_{r \rightarrow r_{+}} 
\frac{\partial f}{\partial \theta}(r,\theta) = 0  
\end{equation*}
for all $\theta \in I_{\theta}$.
\end{itemize}
\item[(ii)] For every $f \in D(A_0)$, 
we define 
\begin{equation*}
A_0 f := D^2_{r \theta} f \, \, . 
\end{equation*}
\end{itemize}
\end{definition}

\begin{lemma} \label{A0isself-adjoint}
$A_0$ is a densely-defined, linear, symmetric 
and essentially self-adjoint operator 
in $X$. In addition, the closure ${\bar{A}}_0$ of 
$A_0$ is semibounded with lower bound 
\begin{equation*}
\alpha := - \frac{m^2 a^2}{4 M^2 r_{+}^2} \, \, .
\end{equation*} 
\end{lemma}

\begin{proof}
See Lemma~$2$ and Theorem~$4$ in \cite{beyer3}.
\end{proof}

\begin{theorem} \label{core}
The span, $D$, of all products 
\begin{equation*}
f \otimes 
(P^{m}_{l} \circ \cos) \, \, ,
\end{equation*}
where $f \in C^{2}_{0}((r_{+},\infty),{\mathbb{C}})$
and $P^{m}_{l} : (-1,1) \rightarrow {\mathbb{R}}$
is the generalized Legendre polynomial corresponding 
to $m \in {\mathbb{Z}}$ and $l \in \{|m|, |m|+1, \dots\}$,  
is a core for ${\bar{A}}_0$.
\end{theorem}

\begin{proof}
For this, we use the notation of \cite{beyer3}. According to 
the proof 
of Theorem~4 of \cite{beyer3}, the underlying 
sets of $X$ and ${\bar{X}} := L^2({\Omega}_s,(r^4/\Delta)\sin \theta))$ are equal; 
and the norms induced on the common set are equivalent, 
the maximal multiplication operator $T_{r^4/(\Delta \overline{\Sigma})}$
by the function $r^4 / (\Delta \overline{\Sigma})$ is a bijective 
bounded linear operator on $X$ that has a bounded linear inverse;
the operator $H$, related to $A_{0}$ by  
\begin{equation} \label{relationA0andH}
A_0 = T_{r^4/(\Delta \overline{\Sigma})} H \, \, ,
\end{equation}
is a densely-defined, linear, symmetric, semibounded and 
essentially self-adjoint 
operator in ${\bar{X}}$, and $D$  is contained in the 
(coinciding) domains of $A_0$ and $H$. Further, it has been 
shown that $(H - \lambda)D$
is dense in ${\bar{X}}$ for 
$\lambda < \beta$, where $\beta := - m^2 a^2/r_{+}^4$ 
is a lower bound 
for $H$. From this follows that $D$ is a core  
for the closure ${\bar{H}}$ of $H$.  
For the proof, let $f \in D(\bar{H})$.
Since $(H - \lambda)D$ is dense in ${\bar{X}}$, there is a sequence 
$f_1,f_2,\dots$ in $D$ such that 
\begin{equation*}
\lim_{\nu \rightarrow \infty} (H - \lambda) f_{\nu} = 
({\bar{H}} - \lambda) f
\, \, .
\end{equation*}
Since ${\bar{H}} - \lambda$ is bijective with a bounded inverse, the 
latter 
implies that $f_1,f_2,\dots$ is convergent to $f$ and also that 
\begin{equation*}
\lim_{\nu \rightarrow \infty} H f_{\nu} = 
{\bar{H}} f \, \, .
\end{equation*} 
Hence, we conclude that ${\bar{H}}$ coincides with the closure 
of $H|_{D}$. Since $T_{r^4/(\Delta \overline{\Sigma})}, 
T_{r^4/(\Delta \overline{\Sigma})}^{-1} \in L(X,X)$, from the latter 
also follows that ${\bar{A}}_{0}$ coincides with the closure 
of $A_{0}|_{D}$. 
\end{proof}

\begin{theorem} \label{friedrichsextension}
The operator ${\bar{A}}_{0}$ coincides with the Friedrichs extension 
of the restriction of $A_0$ to $C^{\infty}_{0}({\Omega}_s,{\mathbb{C}})$.
\end{theorem}

\begin{proof}
As a consequence of Theorem~3 in \cite{beyer1}, it follows that $D$
is contained in the domain of the Friedrichs extension $A_{F}$
of the restriction of $A_0$ to 
$C^{\infty}_{0}({\Omega}_s,{\mathbb{C}})$ and that $A_{F} f = 
A_{0} f$ for every $f \in D$. In this connection, note that 
the addition of a multiple of the identity operator `does not affect'
the Friedrichs extension of an operator.\footnote{I.e., if $A$
is a densely-defined, linear, symmetric and semibounded operator in some 
Hilbert space $X$ and $\gamma \in {\mathbb{R}}$, then the 
Friedrichs extension of $A + \gamma$, $(A + \gamma)_{F}$, and 
the sum of the Friedrichs extension of $A$, $A_F$, and $\gamma$
coincide, $(A + \gamma)_{F} = A_{F} + \gamma$.} Since $D$ is a core for 
${\bar{A}}_{0}$, from this follows that $A_{F} \supset {\bar{A}}_{0}$ 
and hence, since $A_{F}$ is in particular symmetric and 
${\bar{A}}_{0}$ is self-adjoint, that $A_{F} = {\bar{A}}_{0}$. 
\end{proof}

\begin{lemma} \label{A0barplusonefourthofBsquaredispositive}
\begin{equation*}
{\cal A} := {\bar{A}}_0 + (1/4) \, B^2
\end{equation*}
is a densely-defined, linear and positive self-adjoint operator 
in $X$.
\end{lemma}

\begin{proof}
That ${\cal A}$ is a densely-defined, linear 
and self-adjoint operator in $X$ is a consequence of 
Theorem~\ref{A0isself-adjoint} and the Rellich-Kato theorem.
For the latter, see e.g. Theorem~X.12 in \cite{reedsimon}, Vol. II. 
The positivity of ${\cal A}$  
is a simple consequence of the fact that 
\begin{align*} 
& \frac{1}{\overline{\Sigma}} \left( 
- \, \frac{m^2 a^2}{\triangle} + \frac{m^2}{\sin^2 \theta} \right) 
+ 
\frac{1}{4} \, b^2 = 
m^2 \left[ \frac{\triangle - 
a^2 \sin^2 \theta}{\triangle \overline{\Sigma} \, \sin^2 \theta} +  
\frac{4 M^2 a^2 r^2}{(\,\triangle \overline{\Sigma}\,)^2 } \right] 
\\
& = \frac{m^2}{ ( \, \triangle \overline{\Sigma} \, )^2 
\sin^2 \theta} 
\left[ (\triangle - 
a^2 \sin^2 \theta) \, \triangle \overline{\Sigma} 
+ 4 M^2 a^2 r^2 \sin^2\!\theta \right] \\
& = \frac{m^2}{( \, \triangle \overline{\Sigma} \, )^2 
\sin^2 \theta} 
\left\{ \, (\triangle - 
a^2 \sin^2 \theta) \, [\triangle (\Sigma + 2 M r) + 4 M^2 r^2]
+ 4 M^2 a^2 r^2 \sin^2\!\theta \, \right\} \\
& = \frac{m^2}{\triangle {\overline{\Sigma}}^2 
\sin^2 \theta} \left[ \, (\triangle - 
a^2 \sin^2 \theta) \, (\Sigma + 2 M r)
+ 4 M^2 r^2 \right] \\
&  = \frac{m^2}{\triangle {\overline{\Sigma}}^2 
\sin^2 \theta} \left[ \, (\Sigma - 2 M r) \, (\Sigma + 2 M r)
+ 4 M^2 r^2 \right] =  
\frac{m^2 \, \Sigma^2}{\triangle {\overline{\Sigma}}^2 
\sin^2 \theta} \geq 0 \, \, . 
\end{align*}
\end{proof}

\section{Formulation of an Initial Value Problem}
\label{formulation}

In the following, we give an initial value formulation for 
equations of the type of (\ref{kleingordonequation}) whose
possibility is indicated by 
Theorem~$4.11$ in \cite{beyer}, see also Theorem~$5.4.11$
in \cite{beyer2}. Here, we give the details of such formulation,
including abstract energy estimates that provide an independent 
basis for the estimate 
(\ref{massineq}) and also for 
its improvement (\ref{newmassinequality}) below.
Specialization of the abstract formulation
to $X$ given by 
(\ref{definitionofX}), $A: = {\bar{A}}_0 - C$,
$B$ given by (\ref{definitionofB}) and $C :=  
- (\alpha + \varepsilon)$ for some $\varepsilon > 0$, 
provides an initial-value 
formulation for (\ref{kleingordonequation}) on every 
open interval $I$ of ${\mathbb{R}}$
along with quantities that 
are conserved under time evolution. Note that in this case
$A + C = {\bar{A}}_0$.
For convenience, the proofs
of the following statements
are given in the Appendix~$2$.  

\begin{ass} \label{XA}
In the following, let $(X,\braket{\,|\,})$ be a non-trivial complex 
Hilbert space and $A$ be a densely-defined, linear and strictly
positive self-adjoint operator in $X$.
\end{ass}

\begin{definition}
We denote by $W^1_{A}$ the complex Hilbert space\footnote{$W^1_{A}$ may 
be regarded as a generalized Sobolev space.}
given by $D(A^{1/2})$ 
equipped with 
the scalar product $\braket{\,|\,}_1$, defined by 
\begin{equation*}
\braket{\xi|\eta}_1 := \braket{A^{1/2}\xi|A^{1/2}\eta} 
+ \braket{\xi|\eta}   
\end{equation*}
for every $\xi,\eta \in D(A^{1/2})$, and induced norm $\|\,\,\|_{1}$.
\end{definition}

\begin{rem} Note that, as a consequence of 
\begin{equation*}
\|\xi\|_{1} = (\|A^{1/2} \xi\|^2 + \|\xi\|^2)^{1/2} \geq \|\xi\|
\end{equation*}
for every $\xi \in D(A^{1/2})$, the imbedding 
$W^1_{A} \hookrightarrow X$
is continuous.
\end{rem}

\begin{ass} \label{BC}
Let $B : D(A^{1/2}) \rightarrow X$ be a symmetric linear operator 
in $X$ for which there are $a \in [0,1)$ and $b \in [0,\infty)$
such that 
\begin{equation*}
\|B \xi\|^{2} \leq a^2 \|A^{1/2} \xi\|^{2} + b^2 \|\xi\|^{2}
\end{equation*} 
for every $\xi  \in D(A^{1/2})$. Note that this implies that 
$B \in L(W^1_{A},X)$. Further, let  
$C \in L(W^1_{A},X)$ be a symmetric linear operator in $X$ and
$I$ be a non-empty open interval of ${\mathbb{R}}$.
\end{ass}

\begin{definition}
We define a solution space $S_{I}$ to consist of all differentiable 
$u : I \rightarrow W^1_{A}$
with ${\textrm{Ran}}(u) \subset D(A)$, such that  
$u^{\prime} : I \rightarrow X$ is differentiable and  
\begin{equation} \label{pde}
(u^{\prime})^{\prime}(t) + i B u^{\prime}(t) + (A + C) u(t) = 0 
\end{equation} 
for every $t \in I$.\footnote{Note that the differentiability of $u$
implies that 
$\textrm{Ran} u^{\prime} \subset W^1_{A}$.}
\end{definition}

Note that (\ref{pde}) contains two types of derivatives. Every first 
derivative of $u$ is to be understood in the sense of derivatives of
$W^1_{A}$-valued functions, whereas
every further derivative is to be understood in the sense of derivatives
of $X$-valued functions. Unless otherwise indicated, this 
convention is also adopted in the subsequent 
part of this section. On the other hand, since the imbedding 
$W^1_{A} \hookrightarrow X$ is continuous, differentiability 
in the sense of $W^1_{A}$-valued functions
also implies differentiability in the sense of $X$-valued 
functions, including agreement 
of the corresponding derivatives. In particular,  
every $u \in S_{I}$
also satisfies the equation 
\begin{equation} \label{pde2}
u^{\prime \prime}(t) + i B u^{\prime}(t) + (A + C) u(t) = 0 
\end{equation} 
for every $t \in I$, where here all derivatives are to be understood 
in the sense of derivatives
of $X$-valued functions. Further, note that the assumptions 
on $C$, in general, do not imply that $A + C$ is self-adjoint.

\begin{rem}
According to Theorem~$4.11$ in \cite{beyer}, see also 
Theorem~$5.4.11$ 
in \cite{beyer2}, for every $t_0 \in I$, $\xi \in D(A)$ and 
$\eta \in W^1_{A}$, there is a uniquely determined 
corresponding $u \in S_{I}$
such that $u(t_0) = \xi$ and
$u^{\, \prime}(t_0) = \eta$. The proof uses methods from the 
theory of semigroups of operators. Independently, the uniqueness 
of such $u$ follows more elementary from energy estimates 
in part~(iii) of the subsequent Lemma~\ref{conservationlaws}. 
\end{rem}

Parts~(i) and (ii) of the subsequent Lemma~\ref{conservationlaws} 
give a ``conserved current'' and a ``conserved energy'', 
respectively, that are 
associated with solutions of (\ref{pde}). Part~(iii) gives associated 
energy estimates, that, in particular, imply the uniqueness
of the initial value problem for (\ref{pde}) stated in (iv).

\begin{lemma} \label{conservationlaws}
Let $u \in S_{I}$ and $t_0 \in I$.
Then the following holds. 

\begin{itemize}
\item[(i)] If $v \in S_{I}$, then 
$j_{u,v} : I \rightarrow {\mathbb{C}}$, defined by 
\begin{equation*}
j_{u,v}(t) := \braket{u(t)|v^{\prime}(t)} - \braket{u^{\prime}(t)|v(t)} + 
i \braket{u(t)|B v(t)}
\end{equation*}
for every $t \in I$, is constant. 
\item[(ii)] The function $E_{u} : I \rightarrow 
{\mathbb{R}}$, defined by   
\begin{equation*}
E_{u}(t) := \|u^{\prime}(t)\|^2 + \braket{u(t)|(A + C)u(t)}
\end{equation*}
for every $t \in I$, is constant. 
\item[(iii)] 
In addition, let
$A + C$ be semibounded with lower bound $\gamma \in {\mathbb{R}}$.
Then
\begin{equation*}
\|u(t_2)\| \leq
\begin{cases} 
[\,\|u(t_1)\| + |E_u|^{1/2}(t_2 - t_1)\,] e^{|\gamma|^{1/2} 
\, (t_2-t_1)} &
\mbox{if } \gamma < 0 \, \, , \\
\|u(t_1)\| + E_u^{1/2}
\, (t_2-t_1) & \mbox{if } \gamma = 0 \, \, , \\
(2 E_{u}/\gamma)^{1/2} \left( 1 - e^{-\gamma^{1/2}(t_2-t_1)} \right)+ 
 \|u(t_1)\| e^{- \gamma^{1/2}(t_2-t_1)} & \mbox{if } \gamma > 0 \, \, ,
\end{cases} 
\end{equation*}   
for $t_1, t_2 \in I$ such that $t_1 \leq t_2$.
\item[(iv)] In addition, let
$A + C$ be semibounded. If $v \in S_{I}$ is such that 
\begin{equation*}
u(t_0) = v(t_0) \, \, , \, \, u^{\prime}(t_0) = v^{\prime}(t_0) \, \, , 
\end{equation*}
then $v = u$.
\end{itemize} 
\end{lemma}

\begin{proof} See Appendix~$2$.
\end{proof}

\begin{figure} 
\centering
\includegraphics[width=7cm,height=7cm]{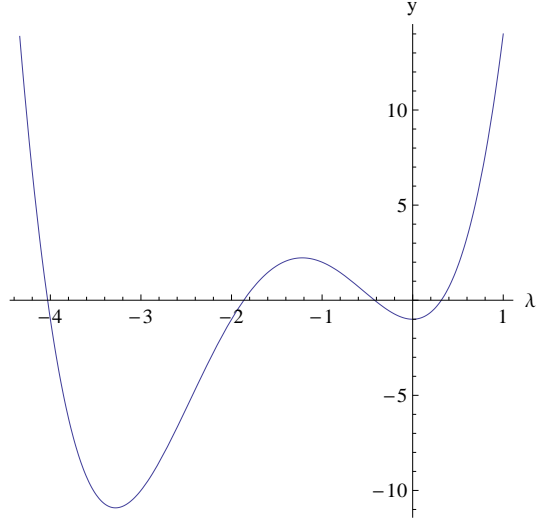}
\caption{Graph of $({\mathbb{R}} \rightarrow L(X,X), \lambda \mapsto 
{\tilde A} - \lambda B - \lambda^2)$ for ${\tilde A}$ and $B$
from Example~\ref{anothercounterexample}.}
\label{fig4}
\end{figure}

The following example proves that it is possible that the 
energy assumes strictly negative values, 
but that the solutions of (\ref{pde}) are stable, i.e., that there 
are no exponentially growing solutions. This is different from the 
case of vanishing $B$, where there are unstable solutions of
(\ref{pde}) if and only if the energy assumes strictly negative 
values.  

\begin{ex} \label{anothercounterexample}
The example uses 
for the Hilbert space $X$ the space
${\mathbb{C}}^2$ equipped with the Euclidean 
scalar product, ${\tilde A}: = A + C$ and $B$ are the linear 
operators on ${\mathbb{C}}^2$ whose representations with respect 
to the canonical basis are given by the matrices
\begin{equation} \label{anotherexample}
\begin{pmatrix}
  1 & 0 \\
  0 & -1
\end{pmatrix} \, \, \textrm{and} \, \, 
\begin{pmatrix}
  3 & 1 \\
  1 & 3
\end{pmatrix} \, \,  , 
\end{equation}
respectively. An analysis shows that ${\tilde A}$ and $B$ are bounded linear and 
self-adjoint operators in $X$, ${\tilde A}$ is semibounded, $B$ is positive
and ${\tilde A} + (1/4) B^2$ is strictly positive. 
Further, ${\tilde A}$ and $B$ do not commute. Finally, the operator 
polynomial $({\mathbb{C}} \rightarrow L(X,X), \lambda \mapsto
{\tilde A} - \lambda B - \lambda^2)$ has $4$ distinct real eigenvalues.
Therefore, in this case, there are no exponentially growing 
solutions of the corresponding equation (\ref{pde}).  
Fig~\ref{fig4} gives the graph of
$p := ({\mathbb{R}} \rightarrow L(X,X), \lambda \mapsto
\det({\tilde A} - \lambda B - \lambda^2)) = \lambda^4 + 
6 \lambda^3 + 8 \lambda^2 - 1$ which suggests that 
there are precisely $4$ distinct real roots. Indeed, we notice that 
\begin{equation*}
p(-5) > 0 \, \, , \, \, p(-4) < 0 \, \, , \, \,
p(-1) > 0 \, \, , \, \,  p(0) < 0 \, \, ,  \, \, p(1) > 0 
\end{equation*} 
and hence that $p$ has real roots in the intervals $(-5,-4)$, 
$(-4,-1)$, $(-1,0)$ and $(0,1)$. In addition, the value of the 
conserved energy $E_u$ corrresponding to the solution $u$ 
of (\ref{pde}) with initial data $u(0) ={\phantom{}}^{t}(0,1)$
and $u^{\, \prime}(0) ={\phantom{}}^{t}(0,0)$ is $< 0$. 
\end{ex}  

There are other possible definitions for the energy that is associated
with solutions of (\ref{pde}). In cases of vanishing $B$, such are 
usually not of further use. In the case of a nonvanishing $B$, 
they can be
useful as is the case for the RKG. In this 
case,  
the positivity of 
$E_{s,u}$ for sufficiently large masses of the field and 
\begin{equation} \label{specialvalue}
s = \frac{ma}{2M r_{+}}
\end{equation} 
provides a basis for (\ref{massineq})
and its 
improvement (\ref{newmassinequality}) below. 

\begin{corollary} \label{otherconservedquantities}
Let $s \in {\mathbb{R}}$ and $u \in S_{I}$. Then, the 
function $E_{s,u} : I \rightarrow 
{\mathbb{R}}$, defined by   
\begin{equation*}
E_{s,u}(t) := \|u^{\prime}(t) + i s u(t) \|^2 + 
\braket{u(t)|(A + C + s (B - s))u(t)}
\end{equation*}
for every $t \in I$, is constant. If
$A + C + s (B - s)$ is additionally semibounded with lower bound 
$\gamma \in {\mathbb{R}}$, then 
\begin{equation*}
\|u(t_2)\| \leq
\begin{cases} 
[\,\|u(t_1)\| + |E_{s,u}|^{1/2}(t_2 - t_1)\,] e^{|\gamma|^{1/2} 
\, (t_2-t_1)} &
\mbox{if } \gamma < 0 \, \, , \\
\|u(t_1)\| + E_{s,u}^{1/2}
\, (t_2-t_1) & \mbox{if } \gamma = 0 \, \, , \\
(2 E_{s,u}/\gamma)^{1/2} \left( 1 - e^{-\gamma^{1/2}(t_2-t_1)} \right)+ 
 \|u(t_1)\| e^{- \gamma^{1/2}(t_2-t_1)} & \mbox{if } \gamma > 0 \, \, ,
\end{cases} 
\end{equation*}   
for $t_1, t_2 \in I$ such that $t_1 \leq t_2$.
\end{corollary}

\begin{proof} See Appendix~$2$.
\end{proof}

\begin{theorem} \label{stability}
If there is $s \in {\mathbb{R}}$ such that  
$A + C + s (B - s)$ is positive, then there are no 
exponentially growing solutions of (\ref{pde}).
\end{theorem}

\begin{proof} The statement is a direct consequence of 
Corollary~\ref{otherconservedquantities} (or Theorem~$4.17$~(ii)
in \cite{beyer}, see also Theorem~$5.4.17$~(ii) in \cite{beyer2}). 

\end{proof}

\begin{ass} \label{AssumptionsForRKG}
In the following, we assume that $X$ is given by 
(\ref{definitionofX}), $A: = {\bar{A}}_0 - C$,
$B$ is given by (\ref{definitionofB}) and $C :=  
- (\alpha + \varepsilon)$ for some $\varepsilon > 0$. 
\end{ass}

Theorem~\ref{stability} leads to an improvement of the estimate 
(\ref{massineq}).

\begin{theorem}
If 
\begin{equation} \label{newmassinequality}
\mu \geqslant \frac{|m|a}{2Mr_{+}} 
\sqrt{1 + \frac{2M}{r_{+}}}
\, \, ,
\end{equation}
then there are no 
exponentially growing solutions of (\ref{pde}).
\end{theorem}

\begin{proof}
Let $s \in {\mathbb{R}}$.
In the following, we estimate ${\bar A}_{0} + s B - s^2$. 
For this, let $f \in D(A_0)$. Then 
\begin{align*}
& (A_0 + s B - s^2) f  =
\frac{1}{\overline{\Sigma}} 
\left( - \frac{\partial}{\partial r} \triangle  \frac{\partial}{\partial r}
- \frac{1}{\sin \theta} \,
\frac{\partial}{\partial \theta} \sin \theta \,\frac{\partial}{\partial \theta}
+ \frac{m^2}{\sin^2 \theta} + V_{s} \right)f 
\, \, ,
\end{align*}
where  
\begin{align*}
V_s & := - \frac{m^2 a^2}{\triangle} + \mu^2 \Sigma + 
s \, \frac{4mMar}{\triangle}  - s^2 \, \overline{\Sigma} \\
& \, \, = 
 - \frac{(2s M r - m a)^2}{\Delta} + (\mu^2 -s^2) \Sigma  - 2 s^2 M r
\, \, .
\end{align*}
First, we note that 
\begin{equation*}
\frac{m^2}{\sin^2 \theta} \geq m^2 \, \, . 
\end{equation*}
In the following, we assume that $s = ma/(2Mr_{+})$. Then  
\begin{equation*}
V_{s1} := - \frac{(2s M r - m a)^2}{\Delta} = - 
\left(\frac{ma}{r_{+}}\right)^2 +
\left(\frac{ma}{r_{+}}\right)^2 \, 
\frac{2 \sqrt{M^2 - a^2}}{r-r_{-}}  \geq - m^2  \, \, .
\end{equation*}
Further, we define 
\begin{equation*}
V_{s2} := (\mu^2 - s^2) \Sigma  - 2 s^2 M r = 
(\mu^2 - s^2) r^2 - 2 s^2 M r + a^2 (\mu^2 - s^2) \cos^2 \theta 
\, \, .
\end{equation*}
If $\mu \geq |s| \cdot [\,1 + (2M/r_{+})\,]^{1/2}$, then
\begin{equation*}
V_{s2} \geq s^2 \, \frac{2M}{r_{+}} \, r^2 
- 2 s^2 M r + a^2 (\mu^2 - s^2) \cos^2 \theta 
\geq a^2 (\mu^2 - s^2) \cos^2 \theta \geq 0 \, \, .
\end{equation*}
As a consequence, 
\begin{equation*}
\frac{1}{\overline{\Sigma}} \left( \frac{m^2}{\sin^2 \theta} +
V_{s} \right) \geq  0 \, \, . 
\end{equation*}
Further, we conclude that 
\begin{align*}
& \braket{f \otimes 
(P^{m}_{l} \circ \cos)\,|\,(A_0 + s B - s^2) (f \otimes 
(P^{m}_{l} \circ \cos))} \\
& \geq  
\int_{\Omega_{s}} \sin \theta \, (f \otimes 
(P^{m}_{l} \circ \cos))^{*}
\left( - \frac{\partial}{\partial r} \triangle  
\frac{\partial}{\partial r} 
- \frac{1}{\sin \theta} \,
\frac{\partial}{\partial \theta} \sin \theta \,\frac{\partial}{\partial \theta} 
  \right) \\
& \qquad \quad (f \otimes 
(P^{m}_{l} \circ \cos)) \, dr d\theta  \geq 0 
\end{align*} 
for every $f \in C^{2}_{0}((r_{+},\infty),{\mathbb{C}})$
and $l \in \{|m|, |m|+1, \dots\}$. Since $D$ is a core for 
${\bar A}_{0}$, 
this implies that
\begin{equation*}
{\bar A}_0 + s B - s^2 \geq 0 \, \, .
\end{equation*}
Hence the statement follows from Theorem~\ref{stability}.
\end{proof}

The following gives a connection of the 
operator ${\bar A}_0 + s B - s^2$, $s \in {\mathbb{R}}$, 
and the Killing field $\partial_{t} + s \partial_{\varphi}$.
The corresponding proof is given in Appendix~$2$. This connection 
sheds light on the previous proof of the positivity of
${\bar A}_0 + s B - s^2$ for $s = ma/(2Mr_{+})$ for sufficiently 
large $\mu$. Differently to $g_{tt}$, the term 
$g(\partial_{t} + s \partial_{\varphi},\partial_{t} 
+ s \partial_{\varphi})$ is positive in a neighbourhood
of the event horizon, but gradually turns negative away from the 
horizon. The latter is compensated by the mass term 
$\mu^2 \rho$ for sufficiently large $\mu$.

\begin{lemma} \label{ConnectionToKillingVf}
Let $s \in {\mathbb{R}}$ and $\xi := \partial_{t} + s 
\partial_{\varphi}$. Then
\begin{align*}
& [\,A_0 + m s B  - (ms)^2\,] f \\  
& = \frac{1}{g^{tt}} \left[\frac{1}{\sqrt{-|g|}} \, \partial_{r} 
\sqrt{-|g|} \, g^{rr} 
\partial_{r} + \frac{1}{\sqrt{-|g|}} \, \partial_{\theta} 
\sqrt{-|g|} \, g^{\theta \theta} \partial_{\theta} \right] f + 
\frac{ m^2 g(\xi,\xi) + \mu^2 \rho}{-g_{\varphi \varphi}} f \, \, . 
\end{align*}
for every $f \in D(A_0)$, where 
\begin{equation*}
\rho := - [\, g_{tt} g_{\varphi \varphi} - (g_{t\varphi})^2\,] 
= \triangle \sin^2\! \theta \, \, . 
\end{equation*}
\end{lemma}

\begin{proof} See Appendix~$2$.
\end{proof}

Subsequently, we rewrite (\ref{pde2}) into an equivalent 
time-dependent wave 
equation that is governed by a family of unitarily equivalent positive 
self-adjoint operators. The latter equation might turn out 
useful for further 
investigation since only self-adjoint operators are involved.
On the other hand, a subsequent example proves that from 
the abstract properties
of this family alone it cannot be concluded that the solutions of the 
equation are stable. 

\begin{lemma} \label{timedep2ndordereq}
Let $B$ be additionally bounded and 
$u \in S_{I}$. Then, $v : I \rightarrow X$ defined 
by 
\begin{equation*}
v(t) := \exp((it/2)B) u(t)
\end{equation*}
for every $t \in I$ is twice differentiable
in the sense of derivatives
of $X$-valued functions and satisfies
\begin{equation} \label{timedep2ndorderequation}
v^{ \prime \prime}(t) + A(t) v(t) = 0 
\end{equation}
for every $t \in I$,
where
\begin{equation} \label{familyofoperators}
A(t) := \exp((it/2)B) 
\bigg(A + C + \frac{1}{4} \, B^2 \bigg) \exp(-(it/2)B) 
\end{equation}
for every $t \in {\mathbb{R}}$.
\end{lemma}

\begin{proof} See Appendix~$2$.
\end{proof}

The previous 
can be used to prove the stability of the solutions of 
(\ref{pde}) in particular cases where the operators 
$A+C$ and $B$ commute. Note that in these cases, 
there is a further 
conserved ``energy'' associated to the solutions of 
(\ref{pde}).

\begin{theorem} \label{caseofcommutingoperators}
If, in addition, 
$A + C$ is self-adjoint and semibounded, 
$B$ is bounded,   
$A + C$ and $B$ commute, i.e., 
\begin{equation*}
B \circ (A+C) \supset (A+C) \circ B \, \, , 
\end{equation*}  
and 
\begin{equation*}
A + C + \frac{1}{4} \, B^2 \, \, ,
\end{equation*}
is positive, then there are no exponentially growing solutions of 
(\ref{pde}).
\end{theorem}

\begin{proof}
The statement is a simple 
consequence of 
Lemma~\ref{timedep2ndordereq} and 
Lemma~\ref{conservationlaws}~(iii).  
\end{proof}

Coming back to the statement of Lemma~\ref{timedep2ndordereq}, 
for every $t \in I$, the corresponding $A(t)$ is a densely-defined, 
linear and self-adjoint operator in $X$, see, e.g., 
Lemma~\ref{elementarylemma}, in the Appendix. In particular, if 
$A + C + (1/4) \, B^2$ is positive, $A(t)$ is positive, too. For instance,
according to 
Lemma~\ref{A0barplusonefourthofBsquaredispositive}, this is true 
in the special case of the Klein-Gordon equation
(\ref{kleingordonequation}).
Hence in such case 
it might be expected that (\ref{timedep2ndorderequation}) 
for $u \in S_{I}$ implies that $\|u\|$ is not exponentially 
growing since this is the case if $A(t) = {\cal A}$ for every
$t \in I$, where ${\cal A}$ is a densely-defined, linear, positive
self-adjoint operator in $X$. In that case,
$u$ is given by 
\begin{equation} \label{rep}
u(t) = \cos((t-t_0){\cal A}^{1/2})u(t_0) + 
\frac{\sin((t-t_0){\cal A}^{1/2})}{{\cal A}^{1/2}} \, u^{\prime}(t_0)
\end{equation} 
for all $t_0, t \in I$, where $\cos((t-t_0){\cal A}^{1/2})$ and
$\sin((t-t_0){\cal A}^{1/2}/{\cal A}^{1/2})$ denote the bounded
linear operators that are associated by the functional calculus 
for ${\cal A}^{1/2}$ to the restriction of 
$\cos((t-t_0).\textrm{id}_{\mathbb{R}})$ and the restriction of
the continuous extension of  
$\sin ((t-t_0).\textrm{id}_{\mathbb{R}})/\textrm{id}_{\mathbb{R}}$
to $[0,\infty)$, 
respectively,
to the spectrum of 
${\cal A}^{1/2}$ \cite{beyer2}. Note that the solutions (\ref{rep})
are in particular bounded if ${\cal A}$ is strictly positive.
Unfortunately, this expectation 
is in general not true. A counterexample can be found already
on the level of finite dimensional Hilbert spaces.

\begin{figure} 
\centering
\includegraphics[width=7cm,height=7cm]{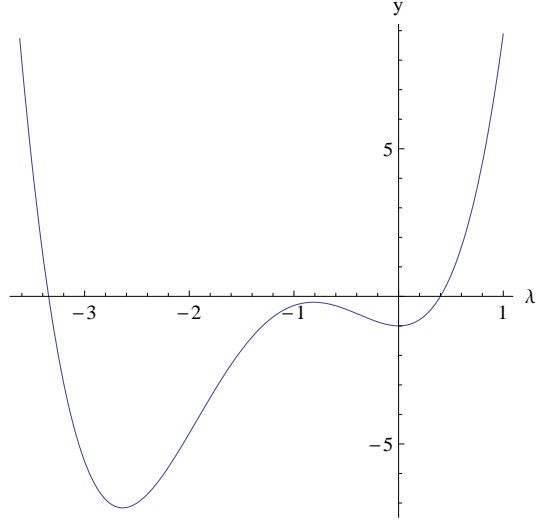}
\caption{Graph of $({\mathbb{R}} \rightarrow L(X,X), \lambda \mapsto
{\tilde A} - \lambda B - \lambda^2)$ for ${\tilde A}$ and $B$
from Example~\ref{counterexample}.}
\label{fig3}
\end{figure}
 
\begin{ex} \label{counterexample}
The example uses 
for the Hilbert space $X$ the space
${\mathbb{C}}^2$ equipped with the Euclidean 
scalar product, ${\tilde A}: = A + C$ and $B$ are the linear 
operators on ${\mathbb{C}}^2$ whose representations with respect 
to the canonical basis are given by the matrices
\begin{equation} \label{example}
\begin{pmatrix}
  1 & 0 \\
  0 & -1
\end{pmatrix} \, \, \textrm{and} \, \, 
\begin{pmatrix}
  23/10 & 1 \\
  1 & 23/10
\end{pmatrix} \, \,  , 
\end{equation}
respectively. An analysis shows that ${\tilde A}$ and $B$ are bounded linear and 
self-adjoint operators in $X$, ${\tilde A}$ is semibounded, $B$ is positive
and ${\tilde A} + (1/4) B^2$ is even 
{\it strictly positive}. 
Further, ${\tilde A}$ and $B$ do not commute. Finally, the operator 
polynomial $({\mathbb{C}} \rightarrow L(X,X), \lambda \mapsto
{\tilde A} - \lambda B - \lambda^2)$ has an eigenvalue with real part
$< 0$. Therefore, in this case, there is an exponentially growing 
solution of the corresponding equation (\ref{pde2}) and hence also 
of (\ref{timedep2ndorderequation}). Note that in this case, 
the corresponding family of operators (\ref{familyofoperators})
consists of strictly positive bounded self-adjoint linear operators whose 
spectra are bounded from below by a common strictly 
positive real number. Fig~\ref{fig3} gives the graph of
$p := ({\mathbb{R}} \rightarrow L(X,X), \lambda \mapsto
\det({\tilde A} - \lambda B - \lambda^2)) = \lambda^4 + 4.6 \lambda^3 + 
4.29 \lambda^2 - 1$ which suggests that 
there are precisely two distinct simple roots. Indeed, this is true. 
The proof proceeds by a discussion of the graph of $p$
using the facts that 
\begin{equation*}
p(-4) > 0 \, \, , \, \, p(-3) < 0 \, \, , \, \,
p(0) < 0 \, \, , \, \,  p(1) > 0 \, \, , 
\end{equation*} 
that 
the zeros of $p^{\,\prime}$ are given by 
\begin{equation*}
(- 69 - \sqrt{1329}\,)/40 \, \, , \, \, 
(-69 + \sqrt{1329}\,)/40 \, \,  , \, \, 0 
\end{equation*}
and that 
\begin{equation*}
p((-69 + \sqrt{1329}\,)/40) < 0 \, \, .
\end{equation*} 
Thus, $({\mathbb{C}} \rightarrow L(X,X), \lambda \mapsto
\det({\tilde A} - \lambda B - \lambda^2))$ has two distinct 
simple real roots and 
a pair of simple complex conjugate roots.
\end{ex}

\section{Discussion}
The mathematical investigation of the stability of 
Kerr space-time has started, but is still in the phase 
of the study of relevant model equations in a Kerr background.
The study of the solutions of the Klein-Gordon equation 
is expected to give important insight into the problem. 
\newline
\linebreak
In the case of 
the wave equation, i.e., for the case of vanishing mass $\mu$ of 
the scalar field,  
results point to the stability of the solutions. On the other
hand, inspection of the reduced Klein-Gordon equation, 
\ref{kleingordonequation}, 
reveals that 
the case of $\mu > 0$ originates from the case $\mu = 0$ by 
the addition of a positive bounded potential term
\begin{equation*} 
\mu^2 \, \frac{\Sigma }{\overline{\Sigma}}
\end{equation*} 
to the equation. If there
were no first order time derivative present in the equation,
from this alone it would be easy to prove 
that the stability of the solutions of the wave equation
implies the stability of the solutions of the 
Klein-Gordon equation for non-vanishing mass. 
\newline
\linebreak
Even in 
the presence of such a derivative, it is hard to believe 
that the addition of 
such term causes instability. In particular, the 
energy estimates in Lemma~\ref{conservationlaws}, 
indicate a stabilizing influence of such a term. 
On the other hand, so far, there is no result that 
would allow to draw such conclusion. 
\newline
\linebreak
The numerical 
results that indicate instability in the case $\mu \neq 0$ 
make quite special assumptions on the values of the rotational 
parameter of the black hole that do not make them look very 
trustworthy. They could very well be numerical 
artefacts. Moreover, the numerical investigation by Konoplya 
et al., \cite{konoplya}, does not find any unstable 
modes and contradicts all these investigations. Also 
the analytical results in this area are not accompanied by 
error estimates and therefore ultimately 
inconclusive. Still, apart from
\cite{zouros},
all these results are consistent with the estimate on $\mu$ 
in \cite{beyer1} and the 
improved 
estimate of this paper, above which 
the solutions of the reduced, by separation in the azimuth angle in 
Boyer-Lindquist coordinates, 
Klein-Gordon equation are stable.
\newline
\linebreak
It seems that the proof of the stability of the solutions 
of the wave equation in a Kerr background will soon be established.
The question of the stability of the massive scalar field 
in a Kerr background is still an open problem, with only few 
rigorous results available, and displays 
surprising mathematical subtlety. 
In particular, in this case 
standard 
tools of theoretical physical investigation, including 
numerical investigations, seem too imprecise for analysis.
Hence a rigorous mathematical investigation, like the one performed 
in this paper, seems to be 
enforced. 

\section{Appendix 1}

In the following, we give the proofs of the 
Lemmatas~\ref{killingfields} and \ref{timelikeregion} from 
Section~\ref{geometricalsetting}.
\newline
\linebreak
Proof of Lemma~\ref{killingfields}.

\begin{proof}
For this, let $s \in {\mathbb{R}}$. Then 
\begin{align*}
& g(\partial_{t} + s \, \partial_{\varphi},
\partial_{t} + s \,  \partial_{\varphi}) = 
g_{tt} + 2s \, g_{t\varphi} + s^2 \, 
g_{\varphi \varphi} \\
& =
1 - \frac{2Mr}{\Sigma} + 4 s \, 
\frac{Mar\sin^2\!\theta}{\Sigma} - s^2 \, 
\frac{\Delta \overline{\Sigma}}{\Sigma} 
\sin^2\!\theta \\
& = \frac{\Delta}{\Sigma} +
\frac{\sin^2\!\theta}{\Sigma} \left[
- a^2 + 4sMar - s^2 
(r^2 + a^2)^{2} + s^2 a^2 \Delta \sin^2\!\theta  
\right] \\
& =  \frac{\Delta}{\Sigma} +
\frac{\sin^2\!\theta}{\Sigma} \left[- (a - 2sMr)^2 + 4 s^2 M^2 r^2
- s^2 
(r^2 + a^2)^{2} + a^2 s^2 \Delta \sin^2\!\theta \, 
\right] \\
& =  \frac{\Delta}{\Sigma} +
\frac{\sin^2\!\theta}{\Sigma} \left[- (a - 2sMr)^2 - 
s^2 \Delta (\Delta + 4 Mr)
+ a^2 s^2 \Delta \sin^2\!\theta \, 
\right] \\
&  =  \frac{\Delta}{\Sigma} -
\frac{\sin^2\!\theta}{\Sigma} \left[(a - 2sMr)^2 + 
s^2 \Delta (\Delta + 4 Mr - a^2 \sin^2\!\theta ) \, 
\right] 
\end{align*}
Hence $g(\partial_{t} + s \, \partial_{\varphi},
\partial_{t} + s \,  \partial_{\varphi})$ has a positive extension 
to the boundary of $\Omega_{s}$ if and only if 
\begin{equation*}
s = \frac{a}{2Mr_{+}} \, \, .
\end{equation*}
In this case, 
\begin{align*}
& (a - 2sMr)^2 + 
s^2 \Delta (\Delta + 4 Mr - a^2 \sin^2\!\theta ) \\
& 
= \frac{a^2}{r_{+}^2} \, (r - r_{+})^2 + 
\frac{a^2}{4M^2r_{+}^2} \,  \Delta (\Delta + 4 Mr - a^2 \sin^2\!\theta ) \\
& = \frac{a^2}{4M^2r_{+}^2} \left[ \, 
4 M^2 (r - r_{+})^2 + \Delta (\Delta + 4 Mr - a^2 \sin^2\!\theta )
\right] \\
& = \frac{a^2 \Delta}{4M^2r_{+}^2} \left[ \, 
4 M^2 \, \frac{r - r_{+}}{r-r_{-}} 
+ \Delta + 4 Mr - a^2 \sin^2\!\theta 
\right] \\
\end{align*}
and hence 
\begin{align*}
& g(\partial_{t} + s \, \partial_{\varphi},
\partial_{t} + s \,  \partial_{\varphi}) \\
& = 
\frac{\Delta}{4M^2r_{+}^2\Sigma} \left[ 4M^2r_{+}^2 
- a^2 \sin^2\!\theta \left( 4 M^2 \, \frac{r - r_{+}}{r-r_{-}} 
+ \Delta + 4 Mr - a^2 \sin^2\!\theta \right) 
\right] \\
& = \frac{\Delta}{4M^2r_{+}^2\Sigma} \left[ 
(2 M r_{+} -  a^2 \sin^2\!\theta)^2 
- a^2 (r - r_{+}) \, \sin^2\!\theta 
\left( \, \frac{4 M^2}{r-r_{-}} 
+ r - r_{-} + 4 M \right) 
\right] \\
& =  \frac{\Delta}{4M^2r_{+}^2\Sigma} \left[ 
(2 M r_{+} -  a^2 \sin^2\!\theta)^2 
- a^2 \Delta \, \sin^2\!\theta 
\left( 1 + \, \frac{2M}{r-r_{-}}\right)^2 
\right]\, \, . 
\end{align*}
\end{proof}

Proof of Lemma~\ref{timelikeregion}.

\begin{proof}
For this, let $(r,\theta) \in \Omega_{e1}$.
Then 
\begin{equation*}
\triangle(r,\theta) <  a^2 \sin^2\!\theta 
\end{equation*}
and
\begin{align*}
& (2 M r_{+} -  a^2 \sin^2\!\theta)^2 
- a^2 \Delta(r,\theta) \, \sin^2\!\theta 
\left( 1 + \, \frac{2M}{r-r_{-}}\right)^2 \\
& =  a^4 \sin^4\!\theta  - \left[ 4 M r_{+} +
 \Delta(r,\theta) 
\left( 1 + \, \frac{2M}{r-r_{-}}\right)^2
\right] a^2 \sin^2\!\theta + 4 M^2 r_{+}^2 \\
& >   (\Delta(r,\theta))^2  - \left[ 4 M r_{+} +
 \Delta(r,\theta) 
\left( 1 + \, \frac{2M}{r-r_{-}}\right)^2
\right] a^2 + 4 M^2 r_{+}^2 \\
& = (\Delta(r,\theta))^2  - a^2
 \Delta(r,\theta)
\left( 1 + \, \frac{2M}{r-r_{-}}\right)^2
 + 4 M r_{+}(M r_{+} - a^2) \\
& = \left[\Delta(r,\theta) - \frac{a^2}{2} \, \left( 1 + \, \frac{2M}{r-r_{-}}\right)^2
\right]^2 -  \frac{a^4}{4} \, \left( 1 + \, \frac{2M}{r-r_{-}}\right)^4
+ 4 M r_{+}(M r_{+} - a^2) \\
& \geq \left[\Delta(r,\theta) - \frac{a^2}{2} \, \left( 1 + \, \frac{2M}{r-r_{-}}\right)^2
\right]^2 + 4 \left[ - \frac{a^4 r_{+}^4}{(r_{+}-r_{-})^4}
+ M r_{+}(M r_{+} - a^2) \right] \, \, .
\end{align*} 
Hence it follows that $(r,\theta) \in \Omega_{e2}$ if 
\begin{align*}
& \frac{a^4 r_{+}^4}{(r_{+}-r_{-})^4} +  a^2 M r_{+} - M^2 r_{+}^2 \\
& 
= \frac{r_{+}^4}{(r_{+}-r_{-})^4} 
\left[ a^4 + \frac{M (r_{+}-r_{-})^4}{r_{+}^3} \, a^2
- \frac{M^2(r_{+}-r_{-})^4}{r_{+}^2}
\right] \leq  0 
\, \, .
\end{align*}
The latter is the case if and only if 
\begin{align*}
a^2 \leq \frac{2Mr_{+}}{1
+ \sqrt{1 + \frac{4r_{+}^4}{(r_{+}-r_{-})^4}}} \, \, .
\end{align*}
Further, 
\begin{align*}
& \frac{2Mr_{+}}{1
+ \sqrt{1 + \frac{4r_{+}^4}{(r_{+}-r_{-})^4}}} \geq 
\frac{Mr_{+}}{1
 + \frac{r_{+}^2}{(r_{+}-r_{-})^2}} \geq  
\frac{M^2}{1
 + \frac{M^2}{M^2 -a^2}} = \frac{M^2(M^2 - a^2)}{2M^2 -a^2} \\
& \geq \frac{1}{2} \, (M^2 - a^2) \, \, .
\end{align*}
Hence if 
\begin{equation*}
a^2 \leq \frac{1}{2} \, (M^2 - a^2) \, \, , 
\end{equation*}
or, equivalently, if condition (\ref{condition}) is satisfied,
it follows that $(r,\theta) \in \Omega_{e2}$. 
\end{proof}

\section{Appendix 2}
In the following, we give the omitted proofs from 
Sections~\ref{basicproperties} and \ref{formulation}.
\newline
\linebreak
Proof of Lemma~\ref{conservationlaws}.

\begin{proof}
`(i)': For this, let $t \in I$ and $h \in {\mathbb{R}}$ such 
that $t + h \in I$. Then 
\begin{align*}
& \frac{j_{u,v}(t+h) - j_{u,v}(t)}{h} \\
& = h^{-1} \left[ 
\braket{u(t+h)|v^{\prime}(t+h)} - \braket{u^{\prime}(t+h)|v(t+h)} + 
i \braket{u(t+h)|B v(t+h)} \right. \\
& \left. \qquad \, \, \, \, \, \, \, \, 
- \braket{u(t)|v^{\prime}(t)} + \braket{u^{\prime}(t)|v(t)} -
i \braket{u(t)|B v(t)} \right] \\
& =  h^{-1}\left[ 
\braket{u(t+h) - u(t)|v^{\prime}(t+h)} + 
\braket{u(t)|v^{\prime}(t+h) - v^{\prime}(t)} \right. \\
& \left.  \qquad\, \, \, \, \, \, \, \,
- \braket{u^{\prime}(t+h)|v(t+h) - v(t)} 
- \braket{u^{\prime}(t+h) - u^{\prime}(t)|v(t)} \right. \\
& \left. \qquad\, \, \, \,  \, \, \, \, +
i \braket{u(t+h) - u(t)|B v(t+h)} + i \braket{B u(t)|v(t+h) -v(t)}
\right] \, \, . 
\end{align*}
Hence it follows that $j_{u,v}$ is differentiable in $t$ with derivative 
\begin{align*}
j_{u,v}^{\prime}(t) & = \braket{u(t)|(v^{\prime})^{\prime}(t)} -
 \braket{(u^{\prime})^{\prime}(t)|v(t)} + 
i \braket{u^{\prime}(t)|B v(t)} + i \braket{B u(t)|v^{\prime}(t)} \\
&  = \braket{u(t)|(v^{\prime})^{\prime}(t) + i B v^{\prime}(t)} -
 \braket{(u^{\prime})^{\prime}(t)+i B u^{\prime}(t)|v(t)} \\
& = - \braket{u(t)|(A+C)v(t)} + \braket{(A+C)u(t)|v(t)}= 0 \, \, .
\end{align*}
From the latter, we conclude that the derivative of 
$j_{u,v}$ vanishes and hence that  
$j_{u,v}$ is a constant function. 
\newline
`(ii)': For this, again, let $t \in I$ and $h \in {\mathbb{R}}$ such 
that $t + h \in I$. Further, let $\tilde{A} := A + C$.  
Then 
\begin{align*}
& \frac{E_{u}(t+h) - E_{u}(t)}{h} \\
& = h^{-1} \left[ 
\braket{u^{\prime}(t+h)|u^{\prime}(t+h)} + \braket{u(t+h)|{\tilde A} u(t+h)}
- \braket{u^{\prime}(t)|u^{\prime}(t)} - \braket{u(t)|{\tilde A} u(t)}
\right] \\
&  = h^{-1} \left[ 
\braket{u^{\prime}(t+h)-u^{\prime}(t)|u^{\prime}(t+h)}  
+ \braket{u^{\prime}(t)|u^{\prime}(t+h)-u^{\prime}(t)} \right. \\
& \left. \qquad\, \, \, \,  \, \, \, \,
+ \braket{u(t+h) - u(t)|{\tilde A} u(t+h)} + 
\braket{u(t)|{\tilde A} (u(t+h)-u(t))} 
\right] \\
&  = h^{-1} \left[ 
\braket{u^{\prime}(t+h)-u^{\prime}(t)|u^{\prime}(t+h)}  
+ \braket{u^{\prime}(t)|u^{\prime}(t+h)-u^{\prime}(t)} \right. \\
& \left. \qquad\, \, \, \,  \, \, \, \,
+ \braket{A^{1/2}(u(t+h) - u(t))|A^{1/2} u(t+h)} + 
\braket{u(t+h) - u(t)|C u(t+h)} \right. \\
& \left. \qquad\, \, \, \,  \, \, \, \, + 
\braket{{\tilde A}u(t)|u(t+h)-u(t)} 
\right] \, \, .
\end{align*}
Hence it follows that $E_{u}$ is differentiable in $t$ with derivative 
\begin{align*}
& \braket{(u^{\prime})^{\prime}(t)|u^{\prime}(t)}  
+ \braket{u^{\prime}(t)|(u^{\prime})^{\prime}(t)}
+ \braket{A^{1/2}u^{\prime}(t)|A^{1/2} u(t)} + 
\braket{u^{\prime}(t)|C u(t)} \\
& + 
\braket{(A+C)u(t)|u^{\prime}(t)} \\
& =  - \braket{i B u^{\prime}(t) + (A + C) u(t)|u^{\prime}(t)}  
- \braket{u^{\prime}(t)|i B u^{\prime}(t) + (A + C) u(t)} \\
& \quad 
+ \braket{u^{\prime}(t)|(A+C)u(t)} + 
\braket{(A+C)u(t)|u^{\prime}(t)} \\
& =  - \braket{i B u^{\prime}(t)|u^{\prime}(t)}  
- \braket{u^{\prime}(t)|i B u^{\prime}(t)} = 0 \, \, . 
\end{align*}
From the latter, we conclude that the derivative of 
$E_{u}$ vanishes and hence that  
$E_{u}$ is a constant function.
\newline
`(iii)': Since 
$A + C$ is semibounded with lower bound $\gamma \in {\mathbb{R}}$,
\begin{equation*}
\braket{\xi|(A + C)\xi} \geq \gamma \|\xi\|^2  
\end{equation*}
for every $\xi \in D(A)$.  Hence it follows by (ii) that 
\begin{equation} \label{energyinequality}
\|u^{\prime}(t)\|^2 + \gamma \|u(t)\|^2 = 
E_{u} - (\braket{u(t)|(A + C)u(t)} - \gamma  \|u(t)\|^2) \leq 
E_{u} 
\end{equation}
for every $t \in {\mathbb{R}}$. 
If $\gamma = 0$, the latter implies that 
\begin{equation*}
\|u^{\prime}(t)\| \leq E_{u}^{1/2} 
\end{equation*}
for every $t \in I$. Hence it follows by weak integration
in $X$, e.g., see Theorem~$3.2.5$ in \cite{beyer2}, that
\begin{equation*}
\|u(t_2) - u(t_1)\| = \left \|\, \int_{(t_1,t_2)} u^{\, \prime}(t) \, dt 
\, \right\|
\leq \int_{(t_1,t_2)} \|u^{\, \prime}(t)\| \, dt \leq  E
_{u}^{1/2}(t_2-t_1) \, \, ,
\end{equation*}
where $t_1, t_2 \in I$ are such that $t_1 < t_2$, 
and hence that 
\begin{equation*}
\|u(t_2)\| \leq \|u(t_1)\| + E_{u}^{1/2}(t_2-t_1) \, \, .
\end{equation*}
For the weak integration, note
that the inclusion of $W^1_{A}$ into $X$ is continuous. 
If $\gamma > 0$, it follows from (\ref{energyinequality}) along with 
the parallelogram identity for elements of $X$
that
\begin{align*}
\|e^{-\gamma^{1/2}t}
(e^{\gamma^{1/2} .{\textrm{id}}_{\mathbb{R}}}
.u)^{\prime}(t)\|^2 = 
\|u^{\prime}(t) + \gamma^{1/2} u(t)\|^2 
\leq  2 (\, \|u^{\prime}(t)\|^2 + \|\gamma^{1/2} u(t)\|^2 \, ) 
\leq 2 E_{u}  
\end{align*}
and hence that 
\begin{equation*}
\|(e^{\gamma^{1/2} .{\textrm{id}}_{\mathbb{R}}}
.u)^{\prime}(t)\| \leq (2 E_{u})^{1/2} 
e^{\gamma^{1/2} t}
\end{equation*}
for $t \in I$.
Hence it follows by weak integration in 
$X$ that 
\begin{align*}
& \|e^{\gamma^{1/2}t_2} 
u(t_2) - e^{\gamma^{1/2}t_1} u(t_1)\| 
= \left\|\int_{(t_1,t_2)} (e^{\gamma^{1/2}.{\textrm{id}}_{\mathbb{R}}}
.u)^{\prime}(t) \, dt
\right\| \\
& \leq \int_{(t_1,t_2)} \|(e^{\gamma^{1/2}.{\textrm{id}}_{\mathbb{R}}}
.u)^{\prime}(t)\| \, dt 
\leq (2 E_{u}/\gamma)^{1/2} \left(e^{\gamma^{1/2}t_2} - 
e^{\gamma^{1/2}t_1} \right)
\end{align*}
for all $t_1, t_2 \in I$ such that $t_1 < t_2$. The latter implies 
that
\begin{equation*}
\|e^{\gamma^{1/2}t_2}u(t_2)\| \leq \|e^{\gamma^{1/2}t_1}u(t_1)\| +
(2 E_{u}/\gamma)^{1/2} \left(e^{\gamma^{1/2}t_2} - 
e^{\gamma^{1/2}t_1} \right) \, \, .
\end{equation*} 
Hence
\begin{equation*}
\|u(t_2)\| \leq (2 E_{u}/\gamma)^{1/2} \left(1 - 
e^{-\gamma^{1/2}(t_2 - t_1)} \right) + 
e^{- \gamma^{1/2}(t_2-t_1)} \|u(t_1)\| 
 \, \, .
\end{equation*} 
If $\gamma < 0$, it follows from (\ref{energyinequality})
that 
\begin{equation*}
\|u^{\prime}(t)\|^2 \leq E_u - \gamma  \|u(t)\|^2 \leq
|E_u| + a \, \|u(t)\|^2 \, \, , 
\end{equation*} 
for every $t \in I$, where $a := - \gamma  > 0$.
The latter implies that 
\begin{equation*}
\|u^{\prime}(t)\| \leq  |E_u|^{1/2} + a^{1/2} \, \|u(t)\|
\end{equation*}
for every $t \in I$. Hence it follows by weak integration in 
$X$ that 
\begin{align*}
& \|u(t_2) - u(t_1)\| = \left\|\int_{(t_1,t_2)} u^{\, \prime}(t) \, dt
\right\| \leq \int_{(t_1,t_2)} \|u^{\, \prime}(t)\| \, dt \\
& 
\leq |E_u|^{1/2}(t_2 - t_1) + a^{1/2} \int_{(t_1,t_2)} \|u(t)\| \, dt
\, \, ,   
\end{align*}
where $t_1, t_2 \in I$ are such that $t_1 < t_2$, and 
\begin{equation*}
\|u(t_2)\| \leq  \|u(t_1)\| + 
|E_u|^{1/2}(t_2 - t_1) + a^{1/2} \int_{(t_1,t_2)} \|u(t)\| \, dt \, \, .
\end{equation*}
By help of the generalized Gronwall inequality from 
Lemma 3.1 in \cite{hale}, from the latter we conclude that  
\begin{equation*}
\|u(t_2)\| \leq
[\,\|u(t_1)\| + |E_u|^{1/2}(t_2 - t_1)\,] e^{a^{1/2} (t_2-t_1)}
\end{equation*}   
for $t_1 \in I$ and $t_2 \in I$ such that $t_1 < t_2$.
\newline
`(iv)': For this, we define $w := v - u$. Then $w$ is an element 
of $S_{I}$ such that $w(t_0) = w^{\prime}(t_0) = 0$. 
This implies that  
\begin{equation*}
E_{w}(t) := \|w^{\prime}(t)\|^2 + \braket{w(t)|(A + C)w(t)}
\end{equation*}
for every $t \in I$ is constant of value $0$. Hence 
we conclude from (iii) that $w(t) = 0_{X}$ for all 
$t \in I$ and therefore that $v = u$.
 
\end{proof}

Proof of Corollary~\ref{otherconservedquantities}.

\begin{proof} We 
define $v : I \rightarrow W^1_{A}$ by 
\begin{equation*}
v(t) := e^{ist} u(t)
\end{equation*}
for every $t \in I$. Then $v$ is differentiable 
with $\textrm{Ran} \, v \subset D(A)$ and also
$v^{\prime} : I \rightarrow X$ is differentiable such that 
\begin{equation*}
v^{\prime}(t) =  e^{ist} [u^{\prime}(t) + is u(t)] \, \, , \, \, 
(v^{\prime})^{\prime}(t) =  e^{ist} [(u^{\prime})^{\prime}(t) + 2is 
u^{\prime}(t) - s^2 u(t)]
\end{equation*} 
for every $t \in I$. Further, 
\begin{align*}
& (v^{\prime})^{\prime}(t) + i (B - 2s) v^{\prime}(t) + 
(A + C + s B -s^2) v(t) \\
&  = e^{ist} [(u^{\prime})^{\prime}(t) + 2is 
u^{\prime}(t) - s^2 u(t) + i (B - 2s)(u^{\prime}(t) + is u(t))  \\
& 
\qquad \quad + (A + C + s B -s^2) u(t)] \\
& = e^{ist} [(u^{\prime})^{\prime}(t) + 2is 
u^{\prime}(t) - s^2 u(t) + i B u^{\prime}(t) 
- 2 i s u^{\prime}(t) - s B u(t) + 2 s^2 u(t) \\
& 
\qquad \quad + (A + C + s B -s^2) u(t)] \\
& =  e^{ist} [(u^{\prime})^{\prime}(t) + i B u^{\prime}(t) + 
(A + C) u(t)] = 0 
\end{align*}
for every $t \in I$. Note that $(X,A,B - 2s, C + s B - s^2)$ 
satisfy Assumptions~\ref{XA},~\ref{BC}. 
Hence it follows by Lemma~\ref{conservationlaws}
that the function $E_{v} : I \rightarrow {\mathbb{R}}$, 
defined by   
\begin{align*}
& E_{v}(t) := \|v^{\prime}(t)\|^2 + \braket{v(t)|(A + C + s B - s^2)v(t)} \\
& = \|u^{\prime}(t) + is u(t)\|^2 + \braket{u(t)|(A + C + s B - s^2)u(t)}
\end{align*}
for every $t \in I$, is constant. 
If, in addition,
$A + C + s (B - s)$ is semibounded with lower bound 
$\gamma \in {\mathbb{R}}$, then 
\begin{equation*}
\|v(t_2)\| \leq
\begin{cases} 
[\,\|v(t_1)\| + |E_{v}|^{1/2}(t_2 - t_1)\,] e^{|\gamma|^{1/2} 
\, (t_2-t_1)} &
\mbox{if } \gamma < 0 \, \, , \\
\|v(t_1)\| + E_{v}^{1/2}
\, (t_2-t_1) & \mbox{if } \gamma = 0 \, \, , \\
(2 E_{v}/\gamma)^{1/2} \left( 1 - e^{-\gamma^{1/2}(t_2-t_1)} \right)+ 
 \|v(t_1)\| e^{- \gamma^{1/2}(t_2-t_1)} & \mbox{if } \gamma > 0 \, \, ,
\end{cases} 
\end{equation*}   
for $t_1, t_2 \in I$ such that $t_1 \leq t_2$.
\end{proof}

Proof of Lemma~\ref{ConnectionToKillingVf}.

\begin{proof}
First, we notice that the only non-vanishing components of 
$(g^{ab})_{(a,b) \in \{t,r,\theta,\varphi\}^2}$ are given by 
\begin{align*}
& g^{tt} = \frac{\overline{\Sigma}}{\Sigma} \, \, , \, \, 
g^{t\varphi} = g^{\varphi t} = \frac{2Mar}{\triangle \Sigma} 
\, \, , \, \, g^{rr} = - \frac{\triangle}{\Sigma} \, \, , \, \, 
g^{\theta \theta} =  - \frac{1}{\Sigma} \, \, , \, \, \\
& g^{\varphi \varphi} = - \frac{1}{\triangle \sin^2\!\theta} \, 
\left(1 -\frac{2Mr}{\Sigma} \right) \, \, .
\end{align*}
Further, we notice that 
\begin{align*}
g^{tt} = - \frac{g_{\varphi \varphi}}{\rho} \, \, , \, \, 
g^{t\varphi} =  \frac{g_{t \varphi}}{\rho} \, \, , \, \, 
g^{\varphi \varphi} = - \frac{g_{tt}}{\rho} \, \, , 
\end{align*} 
where 
\begin{equation*}
\rho := - [\, g_{tt} g_{\varphi \varphi} - (g_{t\varphi})^2\,] 
= \triangle \sin^2\! \theta \, \, . 
\end{equation*}
Hence
\begin{align*}
\frac{1}{g^{tt}} \, \Box = \, & \partial_{t}^2 + 
2 \, \frac{g^{t \varphi}}{g^{tt}} \, \partial_{t} \partial_{\varphi}
+ \frac{g^{\varphi \varphi}}{g^{tt}} \, \partial_{\varphi}^2 \\
& +  \frac{1}{g^{tt}} \left[\frac{1}{\sqrt{-|g|}} \, \partial_{r} 
\sqrt{-|g|} \, g^{rr} 
\partial_{r} + \frac{1}{\sqrt{-|g|}} \, \partial_{\theta} 
\sqrt{-|g|} \, g^{\theta \theta} \partial_{\theta} \right]  \\
=  \, & \partial_{t}^2 +
2 \, 
\frac{g_{t \varphi}}{-g_{\varphi \varphi}} \, \partial_{t} \partial_{\varphi}
- \frac{g_{tt}}{-g_{\varphi \varphi}} \, \partial_{\varphi}^2 \\
& +  \frac{1}{g^{tt}} \left[\frac{1}{\sqrt{-|g|}} \, \partial_{r} 
\sqrt{-|g|} \, g^{rr} 
\partial_{r} + \frac{1}{\sqrt{-|g|}} \, \partial_{\theta} 
\sqrt{-|g|} \, g^{\theta \theta} \partial_{\theta} \right] \, \, . 
\end{align*}
As a consequence, 
\begin{equation*}
A_0 f  = \frac{1}{g^{tt}} \left[\frac{1}{\sqrt{-|g|}} \, \partial_{r} 
\sqrt{-|g|} \, g^{rr} 
\partial_{r} + \frac{1}{\sqrt{-|g|}} \, \partial_{\theta} 
\sqrt{-|g|} \, g^{\theta \theta} \partial_{\theta} \right] f + 
\frac{ m^2 g_{tt} + \mu^2 \rho}{-g_{\varphi \varphi}} f
\end{equation*}
for every $f \in D(A_0)$. Finally, it follows that ,  
\begin{align*}
& [\,A_0 + m s B  - (ms)^2\,] f =  
A_0 f + m s \, 2 m \, \frac{g^{t \varphi}}{g^{tt}} f - (ms)^2 f \\
& = \frac{1}{g^{tt}} \left[\frac{1}{\sqrt{-|g|}} \, \partial_{r} 
\sqrt{-|g|} \, g^{rr} 
\partial_{r} + \frac{1}{\sqrt{-|g|}} \, \partial_{\theta} 
\sqrt{-|g|} \, g^{\theta \theta} \partial_{\theta} \right] f \\
& \, \, \quad + 
\frac{m^2}{-g_{\varphi \varphi}} 
\left( g_{tt}  + 
2 s \, g_{t \varphi} + s^2 
g_{\varphi \varphi} \right) f  + \frac{\mu^2 \rho}{-g_{\varphi \varphi}} \\
& = \frac{1}{g^{tt}} \left[\frac{1}{\sqrt{-|g|}} \, \partial_{r} 
\sqrt{-|g|} \, g^{rr} 
\partial_{r} + \frac{1}{\sqrt{-|g|}} \, \partial_{\theta} 
\sqrt{-|g|} \, g^{\theta \theta} \partial_{\theta} \right] f + 
\frac{m^2 g(\xi,\xi) + \mu^2 \rho}{-g_{\varphi \varphi}} f \, \, . 
\end{align*}
for every $f \in D(A_0)$.
\end{proof} 

Proof of Lemma~\ref{timedep2ndordereq}.

\begin{proof}

First, if $D \in L(X,X)$ and 
$f : I \rightarrow X$ is differentiable in $t \in I$ and 
$h \in {\mathbb{R}}^{*}$ such that $t + h \in I$, it follows
that  
\begin{align*}
& \frac{1}{h} \, \left[ \exp((t+h)D) f(t+h)  - 
\exp(tD) f(t) \right] = 
\exp(tD)  \, \frac{1}{h} \left[ \exp(hD) f(t+h)  - 
f(t) \right] \\
& = \exp(tD) \left[ \exp(h D) \, \frac{1}{h}[f(t+h) - f(t)]  
+ \frac{1}{h}(\exp(h D) f(t) - f(t)) \right] \\
& = \exp(tD) \left[ \exp(h D) \, 
\left(\frac{1}{h}[f(t+h) - f(t)] - f^{\prime}(t) \right) + 
 \exp(h D)f^{\prime}(t) \right. \\
& \left. \qquad \qquad \quad \,  
+ \frac{1}{h}(\exp(h D) f(t) - f(t)) \right]
\end{align*}
and hence that $g := (I \rightarrow X, s \mapsto \exp(sD) f(s))$ is 
differentiable in $t$ with derivative 
\begin{equation*}
\exp(tD) [ f^{\prime}(t)
+ D f(t)] \, \, .
\end{equation*}
In particular, this implies, if $f$ is twice 
differentiable in $t \in I$, that $g$ is twice differentiable 
in $t$ with second derivative 
\begin{equation*}
\exp(tD) [f^{\prime \prime}(t)
+ 2 D f^{\, \prime}(t) 
+ D^2 f(t)] \, \, .
\end{equation*}
Applying the previous auxiliary result to $D = (i/2)B$
proves that $v$ is twice differentiable. Further, 
from the definition of $v$, it follows that 
\begin{equation*}
u(t) = \exp(-(it/2)B) v(t) \, \, , \, \,
\end{equation*}
for every $t \in I$. 
Application of the auxilary results above to $D = -(i/2)B$ 
leads to
\begin{align*} 
& u^{ \prime}(t) = \exp(-(it/2)B) \bigg( v^{ \prime}(t) - 
\frac{i}{2} \, B \, v(t) 
\bigg) \, \, , \\
& u^{ \prime \prime}(t) = \exp(-(it/2)B) \bigg( v^{ \prime \prime}(t) - 
i B \, v^{ \prime}(t) - \frac{1}{4} \, B^2 v(t) 
\bigg) \, \, .
\end{align*}
Hence it follows from (\ref{pde2}) that 
\begin{align*}
& 0 = u^{ \prime \prime}(t) + i B  u^{ \prime}(t)
+ {\tilde A} u(t) \\
& =  \exp(-(it/2)B) \bigg( v^{ \prime \prime}(t) - 
i B \, v^{ \prime}(t) - \frac{1}{4} \, B^2 v(t) + 
i B v^{ \prime}(t) - i B  
\frac{i}{2} \, B \, v(t) \\
& \qquad \qquad \qquad \qquad \,  
+ \exp((it/2)B) {\tilde A} \exp(-(it/2)B) v(t) 
\bigg) \\
& = \exp(-(it/2)B) \bigg( v^{ \prime \prime}(t) 
+ \frac{1}{4} \, B^2 v(t) + \exp((it/2)B){\tilde A} \exp(-(it/2)B) v(t)
\bigg) \\
& = \exp(-(it/2)B) \bigg[ v^{ \prime \prime}(t) 
+ \exp((it/2)B) 
\bigg({\tilde A} + \frac{1}{4} \, B^2 \bigg) \exp(-(it/2)B) v(t)
\bigg] \, \, , 
\end{align*}
where ${\tilde A} := A + C$. 
\end{proof} 

In the following, we give some abstract lemmatas that are 
applied in the text. For the convenience of the reader,
corresponding proofs are added.  

\begin{lemma} \label{elementarylemma}
Let $(X,\braket{\,|\,})$ be a Hilbert space over 
${\mathbb{K}} \in \{{\mathbb{R}},{\mathbb{C}}\}$, 
$A$ a densely-defined, linear and self-adjoint operator 
in $X$ and $U \in L(X,X)$ be unitary. Then, 
$A_{U} := 
U \circ A \circ U^{-1}$
is a densely-defined, linear and self-adjoint operator 
in $X$. Further, if $D \leq D(A)$ is a core for $A$, then 
$U(D)$ is a core for $U \circ A \circ U^{-1}$. Also, 
if $A$ is positive, then $U \circ A \circ U^{-1}$ is positive,
too. 
\end{lemma}

\begin{proof}
First, we note that $D(U \circ A \circ U^{-1}) = U(D(A))$.
Since $D(A)$ is dense in $X$, for $\xi \in X$, there is a sequence 
of $\xi_1,\xi_2,\dots$ of elements of $D(A)$ such that 
\begin{equation*}
\lim_{\nu \rightarrow \infty} \xi_{\nu} = U^{-1} \xi \, \, . 
\end{equation*} 
Hence also 
\begin{equation*}
\lim_{\nu \rightarrow \infty} U \xi_{\nu} = \xi \, \, . 
\end{equation*} 
As a consequence, $U \circ A \circ U^{-1}$ is densely-defined. Also, 
as composition of linear maps, $U \circ A \circ U^{-1}$ is linear.
In addition, for $\xi, \eta \in D(A)$, it follows that 
\begin{equation*}
\braket{U \xi|U \circ A \circ U^{-1}U \eta} = 
\braket{\xi|A \eta} = \braket{A \xi|\eta} =
\braket{U \circ A \circ U^{-1}U \xi|U \eta} 
\end{equation*} 
and hence that $U \circ A \circ U^{-1}$ is symmetric. 
Further, if $\xi \in D((U \circ A \circ U^{-1})^{*})$, then
\begin{equation*}
\braket{(U \circ A \circ U^{-1})^{*}\xi|U \eta} = 
\braket{\xi|(U \circ A \circ U^{-1}) U \eta} =
\braket{U^{-1} \xi | A \eta} 
\end{equation*}
for every $\eta \in D(A)$. Hence $\xi \in U(D(A))$, 
and 
\begin{equation*}
\braket{U^{-1} \xi | A \eta} = \braket{ A U^{-1} \xi | \eta}
=  \braket{U A U^{-1} \xi | U \eta}
\end{equation*} 
for every $\eta \in D(A)$. Since $U(D(A))$ is dense in $X$, this implies
that $(U \circ A \circ U^{-1})^{*} \xi = U A U^{-1} \xi$. 
As a consequence, 
\begin{equation*}
U A U^{-1} \supset (U \circ A \circ U^{-1})^{*} \, \, . 
\end{equation*} 
Hence it follows that 
$U \circ A \circ U^{-1}$ is self-adjoint. Further, let $D \leq D(A)$ 
be a core 
for $A$. As a consequence, for every $\xi \in D(A)$ there is a 
sequence $\xi_{1},\xi_{2},\dots$ in $D$ such that 
\begin{equation*}
\lim_{\nu \rightarrow \infty} \xi_{\nu} = \xi \, \, , \, \, 
\lim_{\nu \rightarrow \infty} A \xi_{\nu} = A \xi \, \, .
\end{equation*} 
Hence $U \xi_{1},U \xi_{2},\dots$ is a sequence in $U(D)$ such that 
\begin{equation*}
\lim_{\nu \rightarrow \infty} U \xi_{\nu} = U \xi \, \, , \, \, 
\lim_{\nu \rightarrow \infty} U A U^{-1} U \xi_{\nu} = U A U^{-1} 
U \xi \, \, .
\end{equation*}
Therefore, $U(D)$ is a core for $U A U^{-1}$. Finally, if $A$ is positive,
it follows for $\xi \in D(A)$ that 
\begin{equation*}
\braket{U \xi | (U \circ A \circ U^{-1}) U \xi} =
\braket{U \xi | U A \xi} = \braket{\xi | A \xi} \geq 0 
\end{equation*}
and hence also the positivity of $U A U^{-1}$.
\end{proof}

\pagebreak

\end{document}